\documentclass[11pt]{article}

\usepackage{newlfont}
\usepackage[below]{placeins}
\usepackage{flafter}
\usepackage{amsfonts}
\usepackage{amsmath}
\usepackage{amssymb}
\usepackage[american]{babel}
\usepackage{url}
\usepackage{graphicx} 
\usepackage{color}
\usepackage{transparent} 

\allowdisplaybreaks

\newcommand{\ket}[1]{\big| #1 \big\rangle}
\newcommand{\bra}[1]{\big\langle #1 \big|}
\newcommand{\braket}[2]{\big\langle #1 \big| #2 \big\rangle}             

\newtheorem{theorem}{Theorem}[section]

\newtheorem{definition}[theorem]{Definition}

\newenvironment{proof}[1][Proof]{\begin{trivlist}
\item[\hskip \labelsep {\bfseries #1}]}{\end{trivlist}}

\newcommand{\qed}{\nobreak \ifvmode \relax \else
      \ifdim\lastskip<1.5em \hskip-\lastskip
      \hskip1.5em plus0em minus0.5em \fi \nobreak
      \vrule height0.75em width0.5em depth0.25em\fi}

\begin{document}

\title{Establishing the equivalence between Szegedy's and \\ coined quantum walks using the staggered model}
\author{Renato Portugal\footnote{portugal@lncc.br} \\
\\
{\small National Laboratory of Scientific Computing - LNCC} \\
{\small Av. Get\'{u}lio Vargas 333, Petr\'{o}polis, RJ, 25651-075, Brazil}
}

\maketitle

\begin{abstract}
Coined Quantum Walks (QWs) are being used in many contexts with the goal of understanding quantum systems and building quantum algorithms for quantum computers. Alternative models such as Szegedy's and continuous-time QWs were proposed taking advantage of the fact that quantum theory seems to allow different quantized versions based on the same classical model, in this case, the classical random walk. In this work, we show the conditions upon which coined QWs are equivalent to Szegedy's QWs. Those QW models have in common a large class of instances, in the sense that the evolution operators are equal when we convert the graph on which the coined QW takes place into a bipartite graph on which Szegedy's QW takes place, and vice versa. We also show that the abstract search algorithm using the coined QW model can be cast into Szegedy's searching framework using bipartite graphs with sinks.
\end{abstract}

\section{Introduction}

The discrete-time coined QW on the line was proposed in the early 1990s in Ref.~\cite{Aharonov:1993}, and is one of the first quantization models of classical random walks. The generalization for regular graphs was proposed in Ref.~\cite{Aharonov:2000}. Early algorithms based on coined QWs with advantage over the classical counterparts were obtained for the element distinctness problem~\cite{Ambainis:2004} and for searching a marked node in a hypercube~\cite{Shenvi:2003}. Many important results were obtained about their asymptotic limit~\cite{Kon02}, localization~\cite{IKK04}, universality~\cite{LCETK10}, and many others as described in reviews~\cite{Ven12,Kon08,Kendon:2007}. Many experimental proposals were described~\cite{Travaglione2002,sanders2003quantum,MPO15}, and experimental implementations were performed~\cite{Karski2009,Zahringer2010,Schreiber2010}.

Using a different quantization procedure, Szegedy~\cite{Szegedy:2004} proposed a new coinless discrete-time QW model on bipartite graphs and was able to provide us with a natural definition of quantum hitting time.  Szegedy also developed QW-based search algorithms, which can detect the presence of a marked vertex at a hitting time that is quadratically smaller than the classical average hitting time on ergodic Markov chains. Szegedy's model was also used for the spatial search problem, that is, for finding the location of a marked vertex in a graph~\cite{Magniez:2011,Krovi:2010}, and for searching triangles~\cite{mss07}.  

The staggered QW model~\cite{PSFG15} plays an important role to connect Szegedy's and coined QWs. Ref.~\cite{Meyer96} analyzed a version of quantum cellular automata that can be converted into a one-dimensional staggered QW, which is equivalent to a generalized version of the coined QW on the line as shown in Refs.~\cite{HKS05,PBF15}. Attempts to obtain a staggered version of QWs for the two-dimensional lattice have appeared in Refs.~\cite{Patel05,Falk:2013}, but Ref.~\cite{PSFG15} showed that the graph considered on those references is a degree-6 crossed lattice, which is not planar. Ref.~\cite{PSFG15} obtained a formulation of staggered QWs on generic graphs, and showed that Szegedy's framework is a subcase of the staggered QW model by using the line graph of the bipartite graph employed in Szegedy's model.


In this work, we characterize which coined QWs can be cast into Szegedy's framework, and which Szegedy's QWs can be converted into the standard coined QW formalism. In the first direction, the shift operator of the coined QW must be Hermitian and the coin must be an orthogonal reflection, which is a unitary and Hermitian operator with special properties in terms of orthogonality of the $(+1)$-eigenvectors. The class of orthogonal reflections includes the Grover and the Hadamard coins. In the other direction, the bipartite graph on which Szegedy's QW takes place must have a special kind of regularity: the degree of the vertices in one of the disjoint sets of vertices must be 2, and the weights associated with the edges incident on those vertices must be equal. Those results show that the Szegedy and the coined QW models share a large class of instances. The staggered QW model bridges the coined and Szegedy's models.

Szegedy's and coined models apparently seem to employ different methods for searching marked vertices. We show that those methods are strongly related. A remarkable searching method based on coined QWs is the abstract search algorithm~\cite{Ambainis:2005,Portugal:book}, which uses coin $(-I)$ on the marked vertices and the Grover coin on the other ones. We show that this method can be cast into Szegedy's searching framework on bipartite graphs with sinks. On the other direction, under some assumptions on the stochastic matrix of the bipartite graph, Szegedy's searching framework can be converted into an equivalent searching method using coined QWs.

The structure of the paper is as follows. In Sec.~\ref{sec:MD}, we present formal definitions of the flip-flop coined QW model on regular and non-regular graphs, Szegedy's QW model on bipartite graphs, the staggered QW model, and important concepts that are employed throughout the work. In Sec.~\ref{sec:MR}, we prove two theorems which connect Szegedy's QWs with coined QWs on regular graphs. In Sec.~\ref{sec:GNR}, we extend the connection for non-regular graphs and give an example. The proofs are left to the Appendix. In Sec.~\ref{sec:searching}, we address the equivalence between the abstract search algorithm using coined QWs and Szegedy's searching framework. In Sec.~\ref{sec:conc}, we draw our conclusions.

\section{Main Definitions}\label{sec:MD}

Let $\Gamma(V,E)$ be a simple graph with vertex set $V$ and edge set $E$ with cardinalities $|V|$ and $|E|$, respectively. We associate the set of vertices, which represents the classical positions, with the vectors of an orthonormal basis of the Hilbert space ${\cal H}^{|V|}$. In the coined QW model on $d$-regular graphs, the total Hilbert space is ${\cal H}^{|V|}\otimes {\cal H}^{d}$, and the walker has $d$ classical directions to move~\cite{Aharonov:2000}. The edges of simple graphs are non-directed, but in many cases we have to consider a non-directed edge equivalent to two superposed opposite arrows.

\begin{definition}\label{def:coinedQW}
The \textbf{standard flip-flop coined QW} on a $d$-regular graph $\Gamma(V,E)$ associated with Hilbert space ${\cal H}^{d\,|V|}$ is driven by a unitary operator the form of which is
\begin{equation}
	U \,=\, S\,(I\otimes C),
\end{equation}
where $C$ is an $d$-dimensional unitary operator (coin), $I$ is the $|V|$-dimensional identity operator, and $S$ is the shift operator which permutes the vectors of the computational basis of ${\cal H}^{d\,|V|}$, and $S^2=I$.
\end{definition}
We make three observations that complement the definition: First, the computational basis of ${\cal H}^{d\,|V|}$ is $\{\ket{v}\ket{j}:\,v\in V,\,0\le j<d\}$, and the action of the \textbf{shift operator} on a vector of the computational basis is
\begin{equation}\label{def_S}
	S \ket{v}\ket{j} = \ket{v'}\ket{j'}, \,\forall v\in V, \,0\le j<d,
\end{equation}
where vertices $v$ and $v'$ are adjacent. If the walker is on vertex $v$, direction $j$ points to $v'$ ($j$ is the label of the directed edge from $v$ to $v'$). In the flip-flop case $(S^2=I)$, we have $S \ket{v'}\ket{j'} = \ket{v}\ket{j}$, which means that if the walker is on vertex $v'$, direction $j'$ points back to $v$ ($j'$ is the label of the directed edge from $v'$ to $v$). Second, in $d$-regular graphs with chromatic index equal to $d$, it  is possible  to find a new shift operator $S'$ similar to $S$ such that $S' \ket{v}\ket{j} = \ket{v'}\ket{j}$, $\forall v\in V$, that is $S'$ does not change the coin value. In this case, if label $j$ points from $v$ to $v'$, the same label $j$ points from $v'$ to $v$. Third, for any discrete-time QW model, if $\ket{\psi_0}$ is the initial state, $U^t\ket{\psi_0}$ is the QW state at step $t$, where $t$ is a non-negative integer. It is interesting to avoid intermediate measurements to take full advantage of the quantum interference.

Definition~\ref{def:coinedQW} does not use the most general shift operator.  However,  the flip-flop shift operator seems to be the most interesting choice for two reasons: First, it is the one that provides the best speedup in spatial search algorithms~\cite{Ambainis:2005}. Second, alternate definitions employed in the literature use information that is external to the graph, such as, go to the right, left, up, or down for the two-dimensional lattice. It is not fair to compare such QW with classical random walks on the same graph because the latter do not use that kind of external information.

The extension of Definition~\ref{def:coinedQW} for non-regular graphs is obtained by noticing that $I\otimes C$ is a direct sum of $|V|$ $d$-dimensional matrices, all of them equal to $C$.
\begin{definition}\label{def:nonregularQW}
 The \textbf{non-regular flip-flop coined QW} on a graph $\Gamma(V,E)$ associated with Hilbert space ${\cal H}^{2\,|E|}$ is driven by a unitary operator the form of which is
\begin{equation}
	U \,=\, S\,C',
\end{equation}
where $C'$ is a direct sum of $|V|$ matrices with dimensions $d_1$, ..., $d_{|V|}$ so that $d_v$ is the degree of vertex $v$, and 
$S$ is the shift operator which permutes the vectors of the computational basis of ${\cal H}^{2\,|E|}$, and $S^2=I$.
\end{definition}
In the non-regular case, the computational basis of ${\cal H}^{2\,|E|}$ is $\{\ket{v,j}: v\in V, \,0\le j<d_v\}$, and the action of the shift operator on a vector of the computational basis is
\begin{equation}\label{def_S2}
	S \ket{v,j} = \ket{v',j'}, \,\forall v\in V, \,0\le j<d_v,
\end{equation}
where vertices $v$ and ${v'}$ are adjacent, label $j$ points from $v$ to ${v'}$, and label $j'$ points from ${v'}$ to $v$. Notice that $\ket{v,j}$ is a notation for the basis vectors that cannot be written as $\ket{v}\otimes\ket{j}$ unless the graph is regular. The order of the basis vectors must be consistent with the fact that $C'$ is a direct sum of $|V|$ matrices. The order is $\ket{v_1,0}$, ...,  $\ket{v_1,d_1-1}$, $\ket{v_2,0}$, ..., $\ket{v_2,d_2-1}$, etc. An example of non-regular flip-flop coined QW is given in Sec.~\ref{sec:GNR} with labels for coin directions and vertices in the graph $\Gamma$ of Fig.~\ref{fig:example6}.

If we associate two different coins (both $d$-dimensional matrices) with two distinct vertices of a $d$-regular graph, the QW is non-regular because $C'$ cannot be factorized and cast into the form $I\otimes C$.

A multigraph $\Gamma(V,E)$ is a generalization of the concept of a simple graph that allows multiple edges between two vertices. Graphs with loops can be added to this class. The results obtained using simple graph can be straightforwardly extended for multigraphs usually overburdening the notation. We will avoid the use of multigraphs whenever possible.

 Let us define an alternate QW model on a bipartite graph known as Szegedy's model~\cite{Szegedy:2004}. Consider a connected bipartite graph $\Gamma(X,Y,E)$, where $X,Y$ are disjoint sets of vertices and $E$ is the set of non-directed edges. Let 
\begin{equation}\label{biadmatrix}
		\left(\begin{array}[]{cc}
		  0 & A \\
			A^T & 0
	\end{array}\right)
\end{equation}
be the biadjacency matrix of $\Gamma(X,Y,E)$. Using $A$, define $P$ as a probabilistic map from $X$ to $Y$ with entries $p_{xy}$. Using $A^T$, define $Q$  as a probabilistic map from $Y$ to $X$ with entries $q_{yx}$. If $P$ is an $m\times n$ matrix, $Q$ will be an $n\times m$ matrix. Both are right-stochastic, that is, each row sums to 1. Using $P$ and $Q$, it is possible to define unit vectors
\begin{eqnarray}
  \ket{\phi_x} &=&  \sum_{y\in Y} \sqrt{p_{x y}}\,\textrm{e}^{i\theta_{xy}} \, \ket{x,y}, \label{ht_phi_x} \\
  \ket{\psi_y}  &=&  \sum_{x\in X} \sqrt{q_{y x}}\,\textrm{e}^{i\theta'_{xy}} \, \ket{x,y}, \label{ht_psi_y}
\end{eqnarray}
that have the following properties: $\braket{\phi_x}{\phi_{x'}}=\delta_{xx'}$ and $\braket{\psi_y}{\psi_{y'}}=\delta_{yy'}$. In Szegedy's original definition, $\theta_{xy}=\theta'_{xy}=0$. We call \textbf{extended Szegedy's QW} the version that allows nonzero angles.

\begin{definition}\label{def:SzegedyQW}
\textbf{Szegedy's QW} on a bipartite graph $\Gamma(X,Y,E)$ with biadjacent matrix (\ref{biadmatrix}) is defined on a Hilbert space ${\cal H}^{m n} = {\cal H}^{m}\otimes {\cal H}^{n} $, where $ m = | X |$ and $n = | Y | $, the computational basis of which is $ \big \{\ket {x, y}: x \in X, y \in Y \big \} $.
The QW is driven by the unitary operator
\begin{equation}\label{ht_U_ev}
    W \,=\, R_1 \, R_0,
\end{equation}
where
\begin{eqnarray}
  R_0 &=& 2\sum_{x\in X} \ket{\phi_x}\bra{\phi_x} - I, \label{ht_RA}\\
  R_1 &=& 2\sum_{y\in Y} \ket{\psi_y}\bra{\psi_y} - I. \label{ht_RB}
\end{eqnarray}
\end{definition}
Notice that operators $R_0$ and $R_1$ are unitary and Hermitian ($R_0^2=R_1^2=I$).

Let ${\cal H}^{|V|}$ be the Hilbert space associated with a graph $\Gamma(V,E)$, the vertices of which are labeled by the vectors of the computational basis. If $U$ is unitary and Hermitian in ${\cal H}^{|V|}$, it can be written as
\begin{equation}
	U \,=\,\sum_x \ket{\psi_x^+}\bra{\psi_x^+} - \sum_y \ket{\psi_y^-}\bra{\psi_y^-},
\end{equation}
where the set of vectors $\ket{\psi_x^+}$ is an orthonormal basis of the $(+1)$-eigenspace, and the set of vectors $\ket{\psi_y^-}$ is an orthonormal basis of the $(-1)$-eigenspace. Using that $\sum_x \ket{\psi_x^+}\bra{\psi_x^+} + \sum_y \ket{\psi_y^-}\bra{\psi_y^-} = I$, we obtain
\begin{equation}\label{U_or}
	U \,=\,2\sum_x\ket{\psi_x^+}\bra{\psi_x^+} - I.
\end{equation}
We want to define a special class of reflection operators $U$ associated with a graph $\Gamma(V,E)$ with the following properties: The $(+1)$-eigenvectors $\ket{\psi_x^+}$ must have non-overlapping nonzero entries, and the sum of those eigenvectors must have no zero entries in the orthonormal basis associated with the vertices of $\Gamma(V,E)$. Each vector $\ket{\psi_x^+}$ forms a clique in $\Gamma(V,E)$ because the vertices associated with nonzero entries of $\ket{\psi_x^+}$ for a fixed $x$ are adjacent. The union of all cliques must be an induced subgraph of $\Gamma(V,E)$. This subgraph is a disconnected union of cliques in general, except when $U$ has only one $(+1)$-eigenvector; in this case, the subgraph and $\Gamma(V,E)$ must be the complete graph.

\begin{definition}\label{def:orthrefl}
A unitary and Hermitian operator $U$ in ${\cal H}^{|V|}$ given by Eq.~(\ref{U_or}) is called an \textbf{orthogonal reflection} of a graph $\Gamma(V,E)$ if there is a complete orthonormal set of $(+1)$-eigenvectors $\ket{\psi_x^+}$ in the orthonormal basis associated with the vertices of the graph obeying the following properties: (1)~if the $i$-th entry of $\ket{\psi_x^+}$ for a fixed $x$ is nonzero, the $i$-th entries of the other $(+1)$-eigenvectors are zero, and (2)~vector $\sum_{x} \ket{\psi_x^+}$ has no zero entries.
\end{definition}
For example, the identity operator $I$ is an orthogonal reflection because the canonical computational basis $\big\{\ket{\psi_x^+}=\ket{x}:\,0\le x<|V|\big\}$ obeys properties (1) and (2). Those $(+1)$-eigenvectors will be used to define a tessellation. $(-I)$ is not an orthogonal reflection because no set of $(+1)$-eigenvectors obeys property (2).

\begin{definition}
A \textbf{polygon} of $\Gamma(V,E)$ induced by vector $\ket{\psi}\in {\cal H}^{|V|}$  is a clique. Two vertices of $\Gamma(V,E)$ are adjacent if the corresponding entries of $\ket{\psi}$ in the basis associated with $\Gamma(V,E)$ are non-zero. A vertex belongs to the polygon if and only if its corresponding entry in $\ket{\psi}$ is non-zero.  An edge belongs to a polygon if and only if the polygon contains the endpoints of the edge. 
\end{definition}

\begin{definition}
A \textbf{tessellation} induced by an orthogonal reflection $U$ of $\Gamma(V,E)$ is the union of the polygons induced by the $(+1)$-eigenvectors $\ket{\psi_x^+}$ of $U$ described in Definition~\ref{def:orthrefl}.
\end{definition}

One of the simplest examples of orthogonal reflection is the Grover operator $G=2\ket{\psi}\bra{\psi}-I$, where $\ket{\psi}$ is the normalized uniform superposition of the vectors of the computational basis~\cite{Portugal:book}. Notice that $G^2=I$, and $G$ has only one eigenvector with eigenvalue $(+1)$, which has no zero entries. The complete graph is induced by $G$ because $\ket{\psi}$ is a superposition of all vertices. If an orthogonal reflection is given, the induced graph can be straightforwardly obtained. If a graph $\Gamma(V,E)$ is given, an orthogonal reflection induces a tessellation of this graph if it contains all necessary edges; the cliques induced by the invariant eigenvectors must be induced subgraphs of $\Gamma(V,E)$. For example, Fig.~\ref{fig:example1} depicts an orthogonal reflection $U$, its corresponding graph $\Gamma_U$, and the induced tessellation in blue. In the general case, polygons of a tessellation do not overlap (property~(1) of Definition~\ref{def:orthrefl}) and a tessellation covers all vertices (property~(2) of Definition~\ref{def:orthrefl}). A tessellation does not need to cover all edges of a predefined graph, unless the graph is induced by an orthogonal reflection as the one in Fig.~\ref{fig:example1}. 
\noindent
\begin{figure}[h!]
\begin{minipage}{0.5\textwidth}
\scriptsize
\begin{equation*}
U\, =\,\frac{1}{3}
  \left[ \begin {array}{ccccc} -1&2&2&0&0\\ \noalign{\medskip}2&-1&2&0&0
     \\ \noalign{\medskip}2&2&-1&0&0\\ \noalign{\medskip}0&0&0&0&3
     \\ \noalign{\medskip}0&0&0&3&0\end {array} \right] 
\end{equation*}
\null
\par\xdef\tpd{\the\prevdepth}
\end{minipage}
\begin{minipage}{0.2\textwidth}
\tiny
\centering 
\def\svgwidth{4cm} 
\begingroup%
  \makeatletter%
  \providecommand\color[2][]{%
    \errmessage{(Inkscape) Color is used for the text in Inkscape, but the package 'color.sty' is not loaded}%
    \renewcommand\color[2][]{}%
  }%
  \providecommand\transparent[1]{%
    \errmessage{(Inkscape) Transparency is used (non-zero) for the text in Inkscape, but the package 'transparent.sty' is not loaded}%
    \renewcommand\transparent[1]{}%
  }%
  \providecommand\rotatebox[2]{#2}%
  \ifx\svgwidth\undefined%
    \setlength{\unitlength}{262.91592239bp}%
    \ifx\svgscale\undefined%
      \relax%
    \else%
      \setlength{\unitlength}{\unitlength * \real{\svgscale}}%
    \fi%
  \else%
    \setlength{\unitlength}{\svgwidth}%
  \fi%
  \global\let\svgwidth\undefined%
  \global\let\svgscale\undefined%
  \makeatother%
  \begin{picture}(1,0.81040071)%
    \put(-0.35222363,0.69291797){\color[rgb]{0,0,0}\makebox(0,0)[lb]{\smash{}}}%
    \put(0,0){\includegraphics[width=\unitlength,page=1]{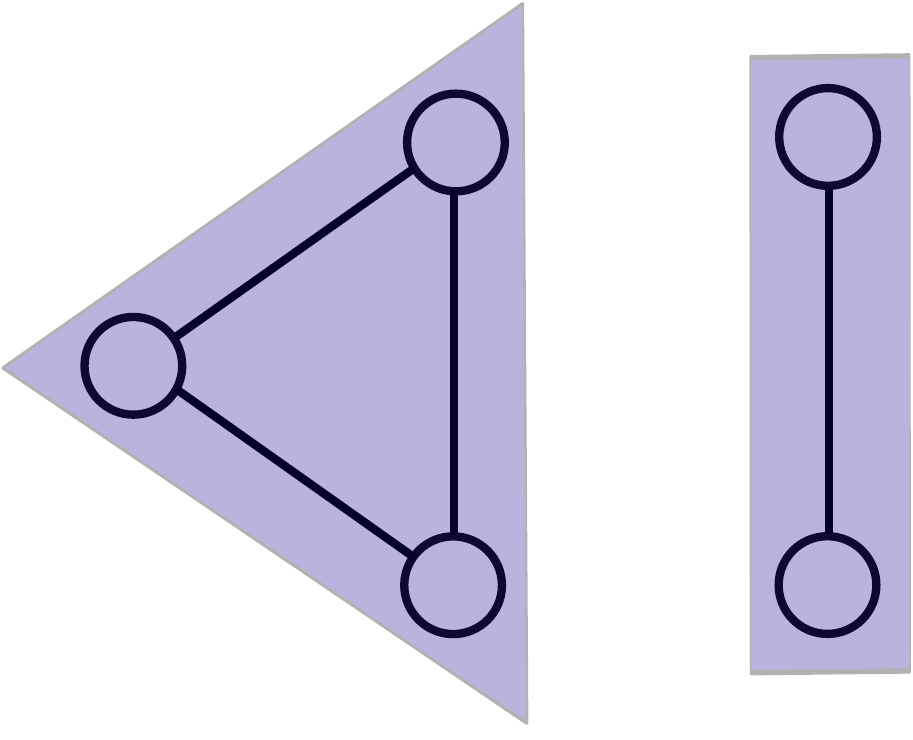}}%
    \put(0.65308426,-0.27072233){\color[rgb]{0,0,0}\transparent{0.08941177}\makebox(0,0)[b]{\smash{}}}%
    \put(0.23165444,0.09836121){\color[rgb]{0,0,0}\transparent{0.87199998}\makebox(0,0)[lt]{\begin{minipage}{0.31744995\unitlength}\centering \end{minipage}}}%
    \put(0.223942,0.3533441){\color[rgb]{0,0,0}\transparent{0.87199998}\makebox(0,0)[b]{\smash{}}}%
    \put(0.14663197,0.39890691){\color[rgb]{0,0,0}\transparent{0.87199998}\makebox(0,0)[b]{\smash{$\ket{0}$}}}%
    \put(0.49806281,0.15888791){\color[rgb]{0,0,0}\transparent{0.87199998}\makebox(0,0)[b]{\smash{$\ket{1}$}}}%
    \put(0.70923404,0.00821294){\color[rgb]{0,0,0}\transparent{0.87199998}\makebox(0,0)[b]{\smash{$\Gamma_U$}}}%
    \put(0.66780108,-0.12655483){\color[rgb]{0,0,0}\transparent{0.87199998}\makebox(0,0)[lt]{\begin{minipage}{0.15743721\unitlength}\centering \end{minipage}}}%
    \put(0.14925315,-0.30303302){\color[rgb]{0,0,0}\transparent{0.87199998}\makebox(0,0)[lt]{\begin{minipage}{0.17704681\unitlength}\centering \end{minipage}}}%
    \put(-0.09217436,0.16640936){\color[rgb]{0,0,0}\transparent{0.87199998}\makebox(0,0)[b]{\smash{}}}%
    \put(0.50128608,0.64020303){\color[rgb]{0,0,0}\transparent{0.87199998}\makebox(0,0)[b]{\smash{$\ket{2}$}}}%
    \put(0.90867493,0.64910376){\color[rgb]{0,0,0}\transparent{0.87199998}\makebox(0,0)[b]{\smash{$\ket{3}$}}}%
    \put(0.9064228,0.15778046){\color[rgb]{0,0,0}\transparent{0.87199998}\makebox(0,0)[b]{\smash{$\ket{4}$}}}%
  \end{picture}%
\endgroup%
\end{minipage}
\caption{Example of an orthogonal reflection $U$ with $(+1)$-eigenvectors $(\ket{0}+\ket{1}+\ket{2})/{\sqrt 3}$ and $(\ket{3}+\ket{4})/{\sqrt 2}$. The corresponding induced subgraph $\Gamma_U$ and its tessellation are depicted on the right-hand side.} 
\label{fig:example1}
\end{figure}
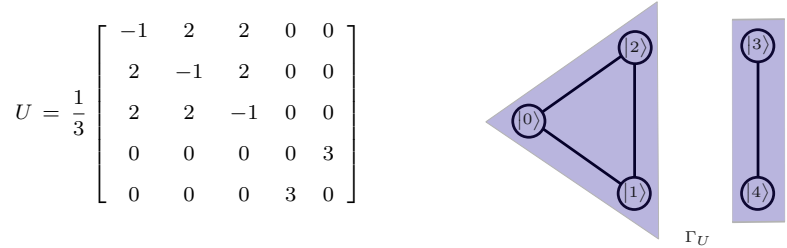

\begin{definition}
The \textbf{staggered QW} on a graph $\Gamma(V,E)$ associated with Hilbert space ${\cal H}^{|V|}$ is driven by 
\begin{equation}
	U \,=\, U_1\,U_0,
\end{equation}
where $U_0$ and $U_1$ are orthogonal reflections of $\Gamma(V,E)$. The union of the tessellations induced by $U_0$ and $U_1$ must cover the edges of $\Gamma(V,E)$.
\end{definition}
The above definition can be readily extended by allowing $U$ to be a product of three or more orthogonal reflections. In this work we focus on the product of only two orthogonal reflections. Ref.~\cite{PSFG15} showed that all Szegedy's QWs are instances of the staggered QW model. In fact, a staggered QW is equivalent to a Szegedy's QW if and only if the intersection of the tessellations induced by $U_0$ and $U_1$ does not contain any edge of  $\Gamma(V,E)$.

To use the staggered QW model for searching marked vertices, we have to use partial tessellations, which can be formally defined by using the notion of partial orthogonal reflection.
\begin{definition}
A unitary and Hermitian operator $U$ in ${\cal H}^{|V|}$ is called a \textbf{partial orthogonal reflection} of a graph $\Gamma(V,E)$ if there is a complete orthonormal set of $(+1)$-eigenvectors $\ket{\psi_x^+}$ in the orthonormal basis associated with the vertices of the graph obeying property~(1) and violating property~(2) of Definition~\ref{def:orthrefl}. 
\end{definition}
A partial orthogonal reflection $U$ induces a \textbf{partial tessellation}, which does not contain all vertices of the graph. The most radical example of a partial orthogonal reflection is the minus identity operator $(-I)\in {\cal H}^N$ because it has no $(+1)$-eigenvectors. The graph induced by this partial orthogonal reflection has $N$ disconnected vertices (the empty $N$-graph). If an $N$-graph is given, a partial orthogonal reflection defines a partial tessellation, which is a tessellation with missing polygons. $(-I)$ induces a partial tessellation with no polygons at all. In the staggered QW model, vertices that do not belong to the intersection of polygons of all tessellations are the marked ones.

\begin{definition}
The \textbf{generalized staggered QW} on a graph $\Gamma(V,E)$ associated with Hilbert space ${\cal H}^{|V|}$ is driven by 
\begin{equation}
	U \,=\, U_1\,U_0,
\end{equation}
where $U_0$ is a partial orthogonal reflection and $U_1$ is an orthogonal reflection of $\Gamma(V,E)$.  The union of the tessellations induced by $U_0$ and $U_1$ must cover the vertices of $\Gamma(V,E)$.
\end{definition}
Again, the above definition can be extended by allowing $U$ to be the product of more (partial) orthogonal reflections.

\section{Results for Regular Graphs}\label{sec:MR}

\begin{theorem}\label{theo1}
A standard flip-flop coined QW on a $d$-regular $N$-graph $\Gamma(V,E)$, such that the coin $C$ is an orthogonal reflection, can be cast into the extended Szegedy's QW model.
\end{theorem}
\begin{proof}
We start by obtaining a staggered QW with two tessellations equivalent to the standard flip-flop coined QW. If $C$ is an orthogonal reflection, then
\begin{eqnarray}
  C &=& 2\sum_{x=0}^{m-1} \ket{\alpha_x}\bra{\alpha_x} - I, \label{C}
\end{eqnarray}
where $\ket{\alpha_0},...,\ket{\alpha_{m-1}}$ is an orthonormal basis for the invariant eigenspace of $C$ with the following properties: (1)~if the $i$-th entry of $\ket{\alpha_x}$ is nonzero, the $i$-th entries of the other $(+1)$-eigenvectors must be zero, and (2)~vector $\sum_{x=0}^{m-1} \ket{\alpha_x}$ has no zero entries. $C$ has an associated $d$-graph $\Gamma_C$, which is a union of $m$ disjoint cliques. The labels of the vertices of $\Gamma_C$ are the coin values.  

Let $\Gamma'(V',E')$ be the graph obtained from $\Gamma(V,E)$ by replacing each vertex $v\in V$ by graph $\Gamma_C$. In the gluing process, a vertex of $\Gamma_C$ with label $j$ is linked by edge $j$ of $\Gamma(V,E)$  ($j$ is the coin direction as in (\ref{def_S})) and receives label $(v,j)$ as a new vertex in $\Gamma'(V',E')$. Fig.~\ref{fig:example2} shows how to obtain $\Gamma'(V',E')$ from the two-dimensional lattice when the coin is the Grover operator. In this example each vertex is replaced by a $4$-clique and two different cliques can have at most one common edge because the two-dimensional lattice is a simple graph.
\begin{figure}[h!] 
\centering
\includegraphics[scale=0.55]{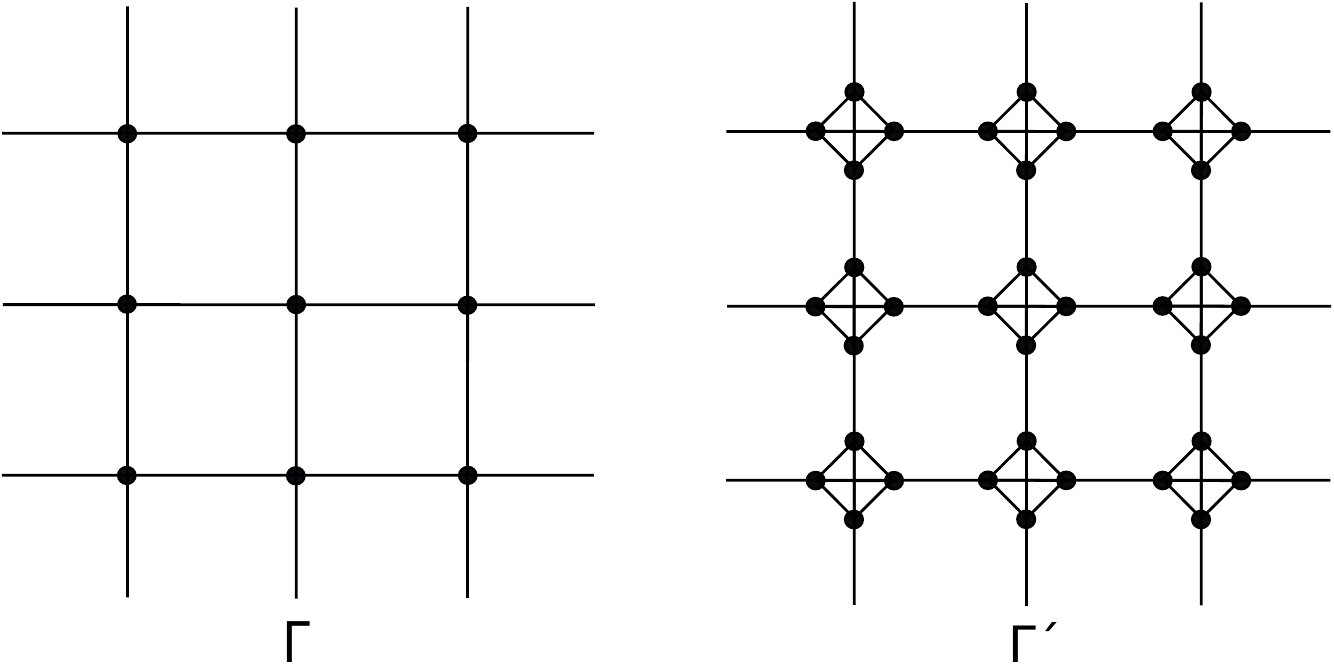}
\caption{A standard flip-flop coined QW with the four-dimensional Grover coin on the two-dimensional lattice $\Gamma$ depicted on the left-hand side is equivalent to a staggered QW on graph $\Gamma'$ on the right-hand side (the vertices are the black circles). Each vertex of the lattice $\Gamma$ is converted into a 4-clique of $\Gamma'$.} 
\label{fig:example2}
\end{figure}

Define
\begin{eqnarray}
  C' &=& I_N\otimes C\nonumber\\
   &=& 2\sum_{v=0}^{N-1}\sum_{x=0}^{m-1} \ket{v,\alpha_x}\bra{v,\alpha_x} - I. 
\end{eqnarray}
Polygons induced by $\ket{v,\alpha_x},\,\forall v,x$ tessellate $\Gamma'(V',E')$ because each polygon induced by $\ket{v,\alpha_x}$ covers exactly graph $\Gamma_C$ that replaces vertex $v$, and the union of polygons induced by $\ket{v,\alpha_x}$ covers all vertices. Fig.~\ref{fig:example3} shows polygons induced by $\ket{v,\alpha_x}$ in blue for the two-dimensional lattice. By analyzing the non-zero entries of vectors $\ket{v,\alpha_x}$, we can verify that $C'$ is an orthogonal reflection of $\Gamma'(V',E')$, which is another way to verify that the set of vectors $\ket{v,\alpha_x}$ induces a tessellation.

\begin{figure}[h!] 
\centering
\includegraphics[scale=0.78]{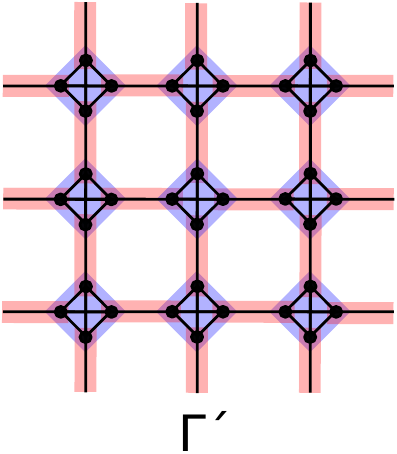}
\caption{Tessellations of a staggered QW equivalent to a standard flip-flop coined QW. Blue polygons are induced by vectors $\ket{v,\alpha_x}$ and define $I\otimes C$;
red polygons are induced by an independent set of vectors $\ket{\beta_{v,j}^+}$ and define $S$. Notice that the red tessellation is a perfect matching; it covers all vertices and has only one edge in each polygon.} 
\label{fig:example3}
\end{figure}

The second tessellation is obtained using the $(+1)$-eigenvectors of $S$, the set of which is a perfect matching of $\Gamma'(V',E')$ as we now show. Using Eq.~(\ref{def_S}), it is straightforward to verify that, for any $v$ and $j$, vectors
\begin{eqnarray}
  \ket{\beta_{v,j}^+} &=& \frac{1}{\sqrt 2}\left(\ket{v}\ket{j}+\ket{v'}\ket{j'}\right),\\
  \ket{\beta_{v,j}^-} &=& \frac{1}{\sqrt 2}\left(\ket{v}\ket{j}-\ket{v'}\ket{j'}\right), 
\end{eqnarray}
are eigenvectors of $S$ with eigenvalues $(+1)$ and $(-1)$, respectively. Since there are $dN/2$ independent eigenvectors associated with each eigenvalue, it follows that
\begin{equation}
	S\,=\, 2 \sum \ket{\beta_{v,j}^+}\bra{\beta_{v,j}^+}-I
\end{equation}
where the sum runs over the set of independent $(+1)$-eigenvectors (the sum has $dN/2$ terms). $S$ is an orthogonal reflection because the set of independent $(+1)$-eigenvectors has non-overlapping nonzero entries, and the sum of those eigenvectors has no zero entries in the computation basis of $\Gamma'(V',E')$. Polygons induced by $\ket{\beta_{v,j}^+}$ cover all vertices and form a perfect matching, which defines a second tessellation of $\Gamma'(V',E')$. Fig.~\ref{fig:example3} shows polygons $\ket{\beta_{v,j}^+}$ in red for the two-dimensional lattice.

The union of tessellations $\ket{v,\alpha_x}$ and $\ket{\beta_{v,j}^+}$ covers all edges and is a well-defined staggered QW having one vertex in each polygon intersection. Using Proposition~4.3 of Ref.~\cite{PSFG15}, this staggered QW can be cast into the extended Szegedy's framework.\qed
\end{proof}

Theorem~\ref{theo1} has a converse, which we state as a new theorem.

\begin{theorem}\label{theo2}
Let $\Gamma(X,Y,E)$ be a biregular bipartite graph such that $\deg(x)=d$, $\forall x\in X$ and $\deg(y)=2$, $\forall y \in Y$. Suppose that if one eliminates the zeros of the sequence $p_{x0},p_{x1},p_{x2},...$ then one gets the same sequence $c_0,c_1,...,c_{d-1}$, for all $x\in X$. Suppose also that $q_{yx}$ is either $1/2$ or $0$. Then Szegedy's QW on $\Gamma(X,Y,E)$ is equivalent to a standard flip-flop coined QW on a $d$-regular $|X|$-multigraph. 
\end{theorem}
\begin{proof}
Consider the staggered QW model on the line graph $L(\Gamma)$ equivalent to Szegedy's QW on $\Gamma(X,Y,E)$~\cite{PSFG15}. $L(\Gamma)$ has $d|X|=2|Y|$ vertices. The polygons of the staggered model are induced by 
\begin{eqnarray}
  \ket{\alpha_x} &=&  \sum_{y\in Y} \sqrt{p_{x y}} \, \ket{f(x,y)}, \label{ht_alpha_x} \\
  \ket{\beta_y}  &=&  \sum_{x\in X} \sqrt{q_{y x}}\, \ket{f(x,y)}, \label{ht_beta_y}
\end{eqnarray} 
where $f$ is the bijection between $E$ and the vertices of $L(\Gamma)$ as described if Ref.~\cite{PSFG15}; and vectors $\ket{\alpha_x},\ket{\beta_x}$ belong to Hilbert space ${\cal H}^{d|X|}$. Using the edge labels described in Fig.~\ref{fig:example5}, vectors $\ket{\alpha_x}$ are given by
\begin{eqnarray}
  \ket{\alpha_x} &=&  \sum_{j=0}^{d-1} \sqrt{c_j} \, \ket{dx+j} \nonumber \\
  &=&   \sum_{j=0}^{d-1} \sqrt{c_j} \, \ket{x}\ket{j}, 
\end{eqnarray}
where vectors $\ket{x}\ket{j}$ belong to the computational basis of Hilbert space ${\cal H}^{|X|}\otimes{\cal H^d}$.

\begin{figure}[h!] 
\centering 
\def\svgwidth{8cm} 
\begingroup%
  \makeatletter%
  \providecommand\color[2][]{%
    \errmessage{(Inkscape) Color is used for the text in Inkscape, but the package 'color.sty' is not loaded}%
    \renewcommand\color[2][]{}%
  }%
  \providecommand\transparent[1]{%
    \errmessage{(Inkscape) Transparency is used (non-zero) for the text in Inkscape, but the package 'transparent.sty' is not loaded}%
    \renewcommand\transparent[1]{}%
  }%
  \providecommand\rotatebox[2]{#2}%
  \ifx\svgwidth\undefined%
    \setlength{\unitlength}{149.14814107bp}%
    \ifx\svgscale\undefined%
      \relax%
    \else%
      \setlength{\unitlength}{\unitlength * \real{\svgscale}}%
    \fi%
  \else%
    \setlength{\unitlength}{\svgwidth}%
  \fi%
  \global\let\svgwidth\undefined%
  \global\let\svgscale\undefined%
  \makeatother%
  \begin{picture}(1,0.56099056)%
    \put(0.37004657,-0.9025818){\color[rgb]{0,0,0}\makebox(0,0)[lb]{\smash{}}}%
    \put(0,0){\includegraphics[width=\unitlength,page=1]{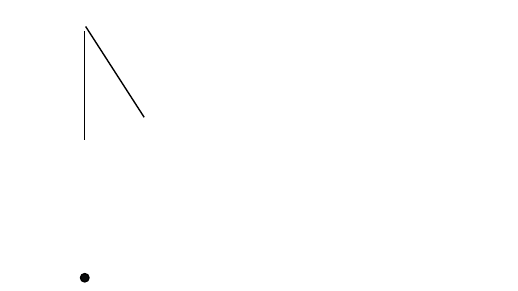}}%
    \put(0.18342718,0.38814226){\color[rgb]{1,0,0}\makebox(0,0)[b]{\smash{}}}%
    \put(0.18398192,0.40201147){\color[rgb]{1,0,0}\makebox(0,0)[b]{\smash{}}}%
    \put(0.21022548,0.33777228){\color[rgb]{0,0,0}\rotatebox{16.44799696}{\makebox(0,0)[b]{\smash{...}}}}%
    \put(0.21782287,0.46580982){\color[rgb]{0,0,0}\makebox(0,0)[b]{\smash{}}}%
    \put(0.29243098,0.41218431){\color[rgb]{0,0,0}\makebox(0,0)[b]{\smash{$d-1$}}}%
    \put(0.13548745,0.35121895){\color[rgb]{0,0,0}\makebox(0,0)[b]{\smash{$0$}}}%
    \put(0,0){\includegraphics[width=\unitlength,page=2]{stag_graph07.pdf}}%
    \put(0.47926331,0.3348846){\color[rgb]{0,0,0}\rotatebox{16.44799696}{\makebox(0,0)[b]{\smash{...}}}}%
    \put(0.56146961,0.40929649){\color[rgb]{0,0,0}\makebox(0,0)[b]{\smash{$2d-1$}}}%
    \put(0.40687977,0.34911569){\color[rgb]{0,0,0}\makebox(0,0)[b]{\smash{$d$}}}%
    \put(0,0){\includegraphics[width=\unitlength,page=3]{stag_graph07.pdf}}%
    \put(0.74881018,0.33309153){\color[rgb]{0,0,0}\rotatebox{16.44799696}{\makebox(0,0)[b]{\smash{...}}}}%
    \put(0.83023259,0.40750376){\color[rgb]{0,0,0}\makebox(0,0)[b]{\smash{$3d-1$}}}%
    \put(0.66177689,0.34365638){\color[rgb]{0,0,0}\makebox(0,0)[b]{\smash{$2d$}}}%
    \put(0.96644339,0.49728305){\color[rgb]{0,0,0}\rotatebox{-0.32181318}{\makebox(0,0)[b]{\smash{...}}}}%
    \put(0.96565845,0.02462645){\color[rgb]{0,0,0}\rotatebox{-0.32181318}{\makebox(0,0)[b]{\smash{...}}}}%
    \put(0,0){\includegraphics[width=\unitlength,page=4]{stag_graph07.pdf}}%
    \put(0.0503575,0.01168443){\color[rgb]{0,0,0}\makebox(0,0)[b]{\smash{$Y$}}}%
    \put(0.0518867,0.50097543){\color[rgb]{0,0,0}\makebox(0,0)[b]{\smash{$X$}}}%
    \put(0.16322145,0.52118214){\color[rgb]{0,0,0}\makebox(0,0)[b]{\smash{$x=0$}}}%
    \put(0.43235567,0.52216084){\color[rgb]{0,0,0}\makebox(0,0)[b]{\smash{$x=1$}}}%
    \put(0.69994826,0.52246744){\color[rgb]{0,0,0}\makebox(0,0)[b]{\smash{$x=2$}}}%
    \put(0,0){\includegraphics[width=\unitlength,page=5]{stag_graph07.pdf}}%
  \end{picture}%
\endgroup%
\caption{Description of the edge labels of the bipartite graph $\Gamma(X,Y,E)$, where $\deg(x)=d$, $\forall x\in X$; $\deg(y)=2$, $\forall y \in Y$. } 
\label{fig:example5}
\end{figure}

Then
\begin{eqnarray}
  U_0 &=&  2\sum_{x\in X} \ket{\alpha_x}\bra{\alpha_x} - I
  \,\,=\,\,   I\otimes C, 
\end{eqnarray}
where
\begin{equation}\label{coin_C_Sz}
C\,=\, 2\ket{\psi}\bra{\psi}-I
\end{equation}
and
\begin{equation}
\ket{\psi}\,=\, \sum_{j=0}^{d-1} \sqrt{c_j} \, \ket{j}.
\end{equation}
Notice that $C\in {\cal H}^d$ is an orthogonal reflection because $C$ has only one eigenvector with eigenvalue $(+1)$ and this eigenvector has no zero entries. The graph induced by $C$ is a $d$-clique.

On the other hand, vectors $\ket{\beta_y}\in {\cal H}^{2|Y|}$ have only two terms
\begin{equation}
\ket{\beta_y}\,=\, \frac{1}{\sqrt 2}\left(\ket{f(x_1,y)}+\ket{f(x_2,y)}\right),
\end{equation}
where $x_1,x_2$ are the neighbors of $y$ ($x_1,x_2$ depend on $y$). Then
\begin{eqnarray}\label{U_1_theo2}
  U_1 &=&  2\sum_{y\in Y} \ket{\beta_y}\bra{\beta_y}- I\nonumber\\
  &=&  \sum_{y\in Y} \,\,\ket{f(x_1,y)}\bra{f(x_2,y)}+\ket{f(x_2,y)}\bra{f(x_1,y)}. \label{U1_theo2}
\end{eqnarray}
$U_1$ is a flip-flop shift operator because $U_1$ commutes basis vectors and $U_1^2=I$. The shift can be understood in the following way: When we convert $\ket{dx_1+j_1}$ into $\ket{x_1}\ket{j_1}$, the interpretation of applying $U_1$ on $\ket{x_1}\ket{j_1}$ is that the walker moves from position $x_1$ in the direction $j_1$ reaching $y$, reflects at $y$, and moves to $x_2$. The state of the walker will be $\ket{x_2}\ket{j_2}$, where $j_2$ points to the same $y$ from $x_2$. Applying $U_1$ on $\ket{x_2}\ket{j_2}$ yields $\ket{x_1}\ket{j_1}$. This inversion of direction characterizes the flip-flop shift operator.

The evolution operator is
\begin{equation}
U \,=\, U_1 \, (I\otimes C),
\end{equation}
where $C$ is the coin operator given by Eq.~(\ref{coin_C_Sz}) and $U_1$ is the flip-flop shift operator given by Eq.~(\ref{U1_theo2}).
Now we have to specify the graph on which the coined QW evolves. The polygons of tessellation $\alpha$ are $d$-cliques, and the polygons of tessellation $\beta$ have two vertices and form a perfect matching of $L(\Gamma)$. Each $d$-clique of $L(\Gamma)$ must be converted into a single vertex. If two $d$-cliques are connected by an edge, the vertices that replace those cliques are adjacent.  If two $d$-cliques are connected by more than one edge, the vertices that replace those cliques must be connected by more than one edge generating a $d$-regular $|X|$-multigraph. Fig.~\ref{fig:example4} shows an example of a bipartite graph $\Gamma$ on which the Szegedy's QW takes place and the multigraph $\Gamma'$ on which an equivalent standard flip-flop coined QW takes place. 
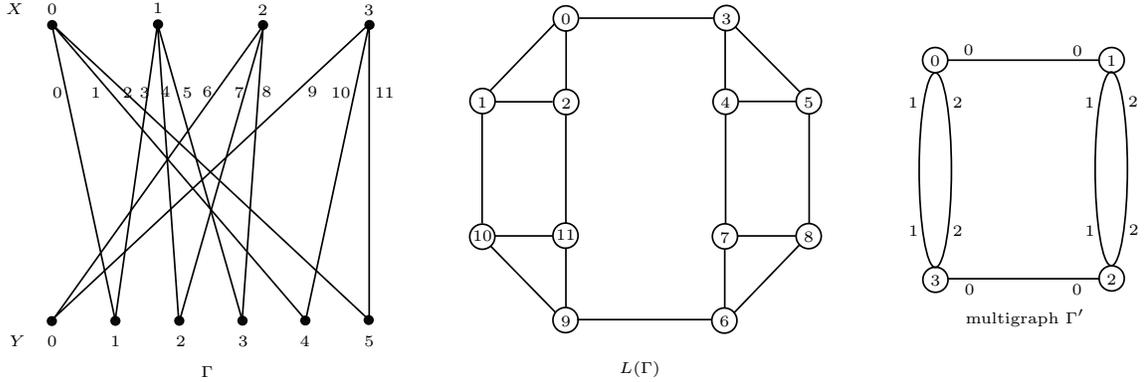
\begin{figure}[h!] 
\tiny
\centering 
\def\svgwidth{15cm} 
\begingroup%
  \makeatletter%
  \providecommand\color[2][]{%
    \errmessage{(Inkscape) Color is used for the text in Inkscape, but the package 'color.sty' is not loaded}%
    \renewcommand\color[2][]{}%
  }%
  \providecommand\transparent[1]{%
    \errmessage{(Inkscape) Transparency is used (non-zero) for the text in Inkscape, but the package 'transparent.sty' is not loaded}%
    \renewcommand\transparent[1]{}%
  }%
  \providecommand\rotatebox[2]{#2}%
  \ifx\svgwidth\undefined%
    \setlength{\unitlength}{538.505561bp}%
    \ifx\svgscale\undefined%
      \relax%
    \else%
      \setlength{\unitlength}{\unitlength * \real{\svgscale}}%
    \fi%
  \else%
    \setlength{\unitlength}{\svgwidth}%
  \fi%
  \global\let\svgwidth\undefined%
  \global\let\svgscale\undefined%
  \makeatother%
  \begin{picture}(1,0.35116924)%
    \put(0,0){\includegraphics[width=\unitlength]{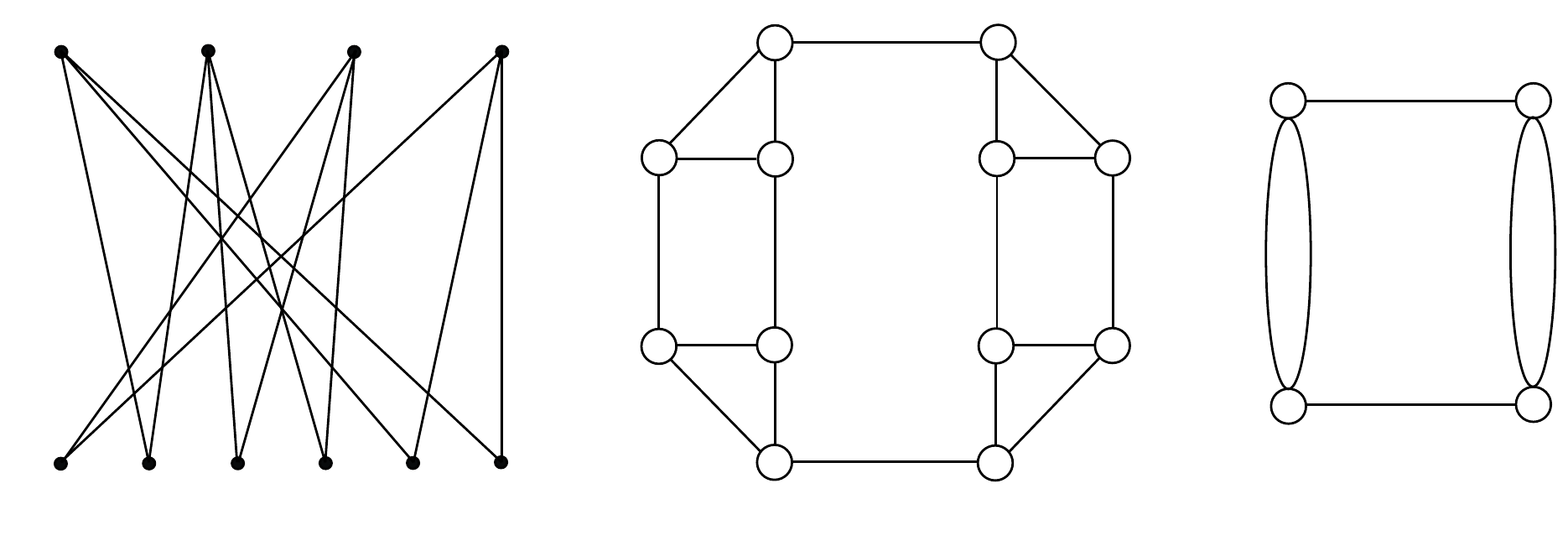}}%
    \put(0.03892959,0.32742776){\color[rgb]{0,0,0}\makebox(0,0)[b]{\smash{0}}}%
    \put(0.13286081,0.32822603){\color[rgb]{0,0,0}\makebox(0,0)[b]{\smash{1}}}%
    \put(0.22564486,0.32678421){\color[rgb]{0,0,0}\makebox(0,0)[b]{\smash{2}}}%
    \put(0.3196458,0.32693892){\color[rgb]{0,0,0}\makebox(0,0)[b]{\smash{3}}}%
    \put(0.03908423,0.03288868){\color[rgb]{0,0,0}\makebox(0,0)[b]{\smash{0}}}%
    \put(0.09514184,0.03337752){\color[rgb]{0,0,0}\makebox(0,0)[b]{\smash{1}}}%
    \put(0.15313001,0.03316297){\color[rgb]{0,0,0}\makebox(0,0)[b]{\smash{2}}}%
    \put(0.20832932,0.03353223){\color[rgb]{0,0,0}\makebox(0,0)[b]{\smash{3}}}%
    \put(0.26315943,0.03300826){\color[rgb]{0,0,0}\makebox(0,0)[b]{\smash{4}}}%
    \put(0.32001493,0.03288868){\color[rgb]{0,0,0}\makebox(0,0)[b]{\smash{5}}}%
    \put(0.04375689,0.25406444){\color[rgb]{0,0,0}\makebox(0,0)[b]{\smash{0}}}%
    \put(0.07789921,0.25363545){\color[rgb]{0,0,0}\makebox(0,0)[b]{\smash{1}}}%
    \put(0.10619,0.25357549){\color[rgb]{0,0,0}\makebox(0,0)[b]{\smash{2}}}%
    \put(0.12093148,0.25400449){\color[rgb]{0,0,0}\makebox(0,0)[b]{\smash{3}}}%
    \put(0.13968882,0.25412429){\color[rgb]{0,0,0}\makebox(0,0)[b]{\smash{4}}}%
    \put(0.15891436,0.25406444){\color[rgb]{0,0,0}\makebox(0,0)[b]{\smash{5}}}%
    \put(0.17650005,0.25443348){\color[rgb]{0,0,0}\makebox(0,0)[b]{\smash{6}}}%
    \put(0.20491031,0.25369518){\color[rgb]{0,0,0}\makebox(0,0)[b]{\smash{7}}}%
    \put(0.22912522,0.25443348){\color[rgb]{0,0,0}\makebox(0,0)[b]{\smash{8}}}%
    \put(0.26991728,0.25437387){\color[rgb]{0,0,0}\makebox(0,0)[b]{\smash{9}}}%
    \put(0.29483555,0.25357549){\color[rgb]{0,0,0}\makebox(0,0)[b]{\smash{10}}}%
    \put(0.33385171,0.25320622){\color[rgb]{0,0,0}\makebox(0,0)[b]{\smash{11}}}%
    \put(0.17104663,0.0061674){\color[rgb]{0,0,0}\makebox(0,0)[lb]{\smash{$\Gamma$}}}%
    \put(0.00114047,0.03294829){\color[rgb]{0,0,0}\makebox(0,0)[lb]{\smash{$Y$}}}%
    \put(-0.00183072,0.32788021){\color[rgb]{0,0,0}\makebox(0,0)[lb]{\smash{$X$}}}%
    \put(0.49437425,0.31876217){\color[rgb]{0,0,0}\makebox(0,0)[b]{\smash{0}}}%
    \put(0.42045849,0.2452189){\color[rgb]{0,0,0}\makebox(0,0)[b]{\smash{1}}}%
    \put(0.63530637,0.12519408){\color[rgb]{0,0,0}\makebox(0,0)[b]{\smash{7}}}%
    \put(0.63588005,0.24463451){\color[rgb]{0,0,0}\makebox(0,0)[b]{\smash{4}}}%
    \put(0.63478466,0.05054686){\color[rgb]{0,0,0}\makebox(0,0)[b]{\smash{6}}}%
    \put(0.63667943,0.31880312){\color[rgb]{0,0,0}\makebox(0,0)[b]{\smash{3}}}%
    \put(0.49456195,0.24435641){\color[rgb]{0,0,0}\makebox(0,0)[b]{\smash{2}}}%
    \put(0.70950497,0.12542644){\color[rgb]{0,0,0}\makebox(0,0)[b]{\smash{8}}}%
    \put(0.70958998,0.24500832){\color[rgb]{0,0,0}\makebox(0,0)[b]{\smash{5}}}%
    \put(0.49437425,0.05112968){\color[rgb]{0,0,0}\makebox(0,0)[b]{\smash{9}}}%
    \put(0.42049505,0.12688956){\color[rgb]{0,0,0}\makebox(0,0)[b]{\smash{10}}}%
    \put(0.49422838,0.12765587){\color[rgb]{0,0,0}\makebox(0,0)[b]{\smash{11}}}%
    \put(0.82161228,0.28156079){\color[rgb]{0,0,0}\makebox(0,0)[b]{\smash{0}}}%
    \put(0.97801168,0.28155344){\color[rgb]{0,0,0}\makebox(0,0)[b]{\smash{1}}}%
    \put(0.82184006,0.08674184){\color[rgb]{0,0,0}\makebox(0,0)[b]{\smash{3}}}%
    \put(0.97805682,0.08792917){\color[rgb]{0,0,0}\makebox(0,0)[b]{\smash{2}}}%
    \put(0.9014139,0.05352108){\color[rgb]{0,0,0}\makebox(0,0)[b]{\smash{multigraph  $\Gamma'$}}}%
    \put(0.80223561,0.24512182){\color[rgb]{0,0,0}\makebox(0,0)[b]{\smash{1}}}%
    \put(0.85119192,0.29113398){\color[rgb]{0,0,0}\makebox(0,0)[b]{\smash{0}}}%
    \put(0.84166234,0.24485052){\color[rgb]{0,0,0}\makebox(0,0)[b]{\smash{2}}}%
    \put(0.95888614,0.131387){\color[rgb]{0,0,0}\makebox(0,0)[b]{\smash{1}}}%
    \put(0.94750573,0.07880839){\color[rgb]{0,0,0}\makebox(0,0)[b]{\smash{0}}}%
    \put(0.99757169,0.13156042){\color[rgb]{0,0,0}\makebox(0,0)[b]{\smash{2}}}%
    \put(0.80246215,0.13113534){\color[rgb]{0,0,0}\makebox(0,0)[b]{\smash{1}}}%
    \put(0.85215969,0.0791497){\color[rgb]{0,0,0}\makebox(0,0)[b]{\smash{0}}}%
    \put(0.84174063,0.13130887){\color[rgb]{0,0,0}\makebox(0,0)[b]{\smash{2}}}%
    \put(0.95873987,0.24524647){\color[rgb]{0,0,0}\makebox(0,0)[b]{\smash{1}}}%
    \put(0.94780418,0.2908139){\color[rgb]{0,0,0}\makebox(0,0)[b]{\smash{0}}}%
    \put(0.99712897,0.24556813){\color[rgb]{0,0,0}\makebox(0,0)[b]{\smash{2}}}%
    \put(0.54191312,0.00958613){\color[rgb]{0,0,0}\makebox(0,0)[lb]{\smash{$L(\Gamma)$}}}%
  \end{picture}%
\endgroup%
\caption{An example of a bipartite graph $\Gamma$, its line graph $L(\Gamma)$, and the reduced multigraph obtained from the line graph by replacing $3$-cliques by single vertices. Szegedy's QW on $\Gamma$ using vectors $\ket{\phi_x}$ and $\ket{\psi_y}$ in uniform superposition is equivalent to a staggered QW on $L(\Gamma)$ using polygons induced by normalized vectors in uniform superposition, which is equivalent to the standard flip-flop coined QW on the multigraph $\Gamma'$ on the right-hand side with the three-dimensional Grover coin.} 
\label{fig:example4}
\end{figure}

We have used the staggered QW version of Szegedy's QW because in the original Szegedy's version there is an idle subspace spanned by the non-edges linking $X$ and $Y$ that hinders the decomposition of $R_0$ given by Eq.~(\ref{ht_RA}) into $I\otimes C$.\qed
\end{proof}

\section{Results for Non-Regular Graphs}\label{sec:GNR}

Theorem~\ref{theo1} can be generalized for non-regular flip-flop coined QWs, and Theorem~\ref{theo2} can be generalized for bipartite graphs that are not biregular. The proofs are given in the Appendix. 

\begin{theorem}\label{theo3}
A non-regular flip-flop coined QW such that $C'$ is an orthogonal reflection can be cast into the extended Szegedy's model.
\end{theorem}

\begin{theorem}\label{theo4}
Let $\Gamma(X,Y,E)$ be a bipartite graph such that $\deg(y)=2$, $\forall y \in Y$. Suppose that $q_{yx}$ is either $1/2$ or $0$. Then Szegedy's QW on $\Gamma(X,Y,E)$ is equivalent to a non-regular flip-flop coined QW on a $|X|$-multigraph. 
\end{theorem}

We give an example that displays the underlying structure of the general proof. Let us start by describing a non-regular flip-flop coined QW on the graph $\Gamma$ depicted on the left-hand side of Fig.~\ref{fig:example6}. Let us use the three-dimensional Grover coin for the vertex of degree 3 and Hadamard coin for the vertices of degree 2.

\begin{figure}[h!] 
\centering
\includegraphics[scale=0.59]{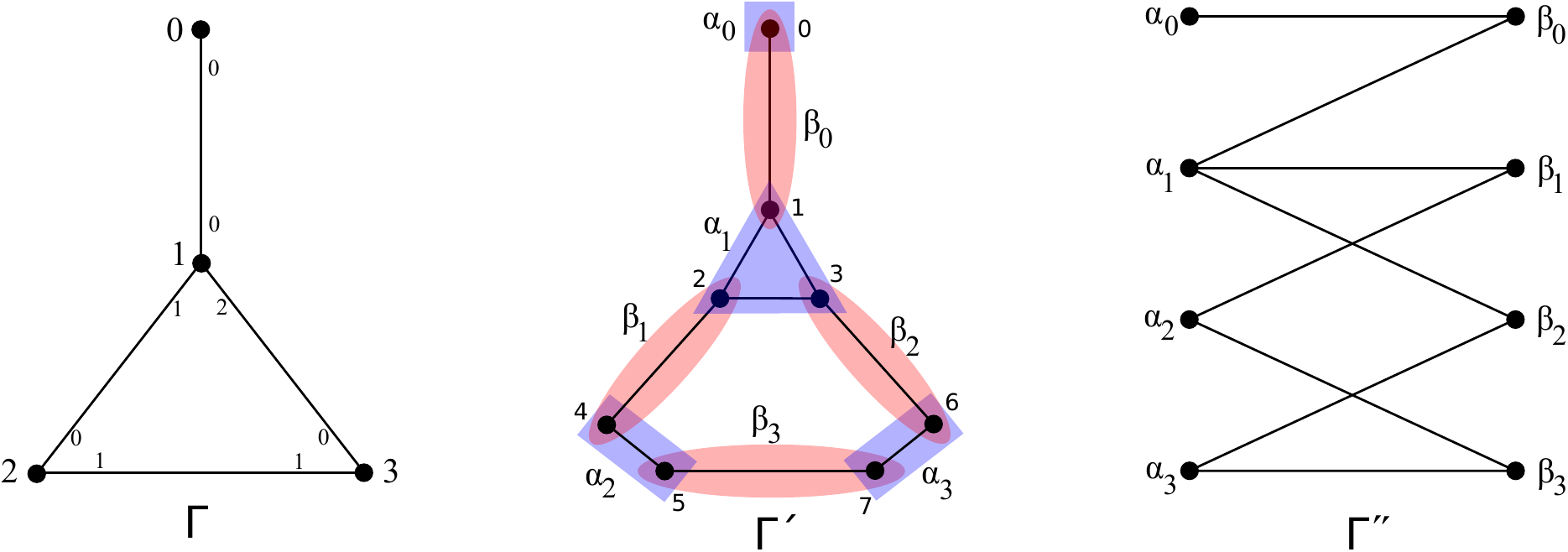}
\caption{Example of a non-regular flip-flop coined QW on graph $\Gamma$ and its equivalent Szegedy's version on the bipartite graph $\Gamma''$. The staggered model on $\Gamma' $ is used as a bridge to go from $\Gamma$ to $\Gamma''$.} 
\label{fig:example6}
\end{figure}

The coin operator is
\begin{equation}
	 C'\,=\,\left[ \begin {array}{cccc} 1&0&0&0\\  0&C&0&0
\\  0&0&H&0\\  0&0&0&H\end {array}
 \right],
\end{equation}
where $H$ is the Hadamard gate and 
\begin{equation}
	 C\,=\, \frac{1}{3}\left[ \begin {array}{ccc} -1&2&2\\  2&-1&2
\\  2&2&-1\end {array} \right].
\end{equation}
The shift operator is 
\begin{equation}
	S\,=\,\left[ \begin {array}{cccccccc} 0&1&0&0&0&0&0&0\\ 1
&0&0&0&0&0&0&0\\  0&0&0&0&1&0&0&0
\\  0&0&0&0&0&0&1&0\\  0&0&1&0&0&0&0
&0\\  0&0&0&0&0&0&0&1\\  0&0&0&1&0&0
&0&0\\  0&0&0&0&0&1&0&0\end {array} \right]. 
\end{equation}

The staggered QW graph equivalent to the coined QW graph is obtained by replacing each vertex of $\Gamma$ by a $d$-clique, where $d$ is the degree of the vertex.  Fig.~\ref{fig:example6} shows the resulting graph $\Gamma'$ with the induced tessellations. Each $d$-clique is a polygon in the tessellation $\alpha$ (blue), and the vertices incident to each edge of the original graph is a polygon of the tessellation $\beta$ (red). Vectors $\ket{\alpha_x}$ and $\ket{\beta_y}$ are given by
\begin{center}
\begin{minipage}{0.4\textwidth}
\begin{eqnarray*}
\ket{\alpha_0}&=&\ket{0}\\
\ket{\alpha_1}&=&\frac{\ket{1}+\ket{2}+\ket{3}}{\sqrt 3}\\
\ket{\alpha_2}&=&\frac{{\sqrt {2+\sqrt {2}}}\,\ket{4}+{\sqrt {2-\sqrt {2}}}\,\ket{5}}{2}\\
\ket{\alpha_3}&=&\frac{{\sqrt {2+\sqrt {2}}}\,\ket{6}+{\sqrt {2-\sqrt {2}}}\,\ket{7}}{2}
\end{eqnarray*}
\mbox{\,}
\end{minipage}
\begin{minipage}{0.4\textwidth}
\begin{eqnarray*}
\ket{\beta_0}&=&\frac{1}{\sqrt 2}\big(\ket{0}+\ket{1}\big)\\
\ket{\beta_1}&=&\frac{1}{\sqrt 2}\big(\ket{2}+\ket{4}\big)\\
\ket{\beta_2}&=&\frac{1}{\sqrt 2}\big(\ket{3}+\ket{6}\big)\\
\ket{\beta_3}&=&\frac{1}{\sqrt 2}\big(\ket{5}+\ket{7}\big)
\end{eqnarray*}
\mbox{\,}
\end{minipage}
\end{center}
where $\ket{\alpha_1}$ is the normalized $(+1)$-eigenvector of the three-dimensional Grover coin, $\ket{\alpha_2}$ and $\ket{\alpha_3}$ are each one the normalized $(+1)$-eigenvector of the Hadarmard gate, and $\ket{\beta_0}$ to $\ket{\beta_3}$ are the $(+1)$-eigenvectors of $S$. It is straightforward to check that 
\begin{eqnarray}
C'&=&2\sum_{j=0}^3 \ket{\alpha_j}\bra{\alpha_j} - I,\\
S &=&2\sum_{j=0}^3 \ket{\beta_j}\bra{\beta_j} - I,
\end{eqnarray}
showing that the staggered QW is equivalent to the coined QW because the evolution operators of those QWs are equal.

Graph $\Gamma''$ on the right-hand side is obtained from $\Gamma'$ by connecting the vertices of $\Gamma''$ associated with overlapping polygons, as explained in Ref.~\cite{PSFG15}. After performing this connecting procedure, $\Gamma'$ is the line graph of $\Gamma''$. Vectors $\ket{\phi_x}$ and $\ket{\psi_y}$ of Szegedy's QW on $\Gamma''$ are obtained from vectors $\ket{\alpha_x}$ and $\ket{\beta_y}$ employing the bijection between the vertices of $\Gamma'$ and the edges of $\Gamma''$: $\ket{0} \leftrightarrow \ket{0, 0}, \ket{1} \leftrightarrow \ket{1, 0}, \ket{2} \leftrightarrow \ket{1, 1}, \ket{3} \leftrightarrow \ket{1, 2}, \ket{4} \leftrightarrow \ket{2, 1}, \ket{5} \leftrightarrow \ket{2, 3}, \ket{6} \leftrightarrow \ket{3, 2}, \ket{7} \leftrightarrow \ket{3, 3}$. Notice that $W$ is in ${\cal H}^{16}$ while $S(I\otimes C)$ is in ${\cal H}^8$. The non-trivial part of $W$ is equal to the evolution operator of the coined QW or the staggered QW.

Since the staggered QW on $\Gamma'$ is equivalent to Szegedy's QW on $\Gamma''$ ~\cite{PSFG15}, it follows that the non-regular flip-flop coined QW on $\Gamma$ is equivalent to Szegedy's QW on $\Gamma''$.

\section{Searching Marked Vertices}\label{sec:searching}

One of the most successful methods to search marked vertices in the coined QW model is to use non-regular flip-flop coined QWs with two different coins: $(-I)$ on the marked vertices and the Grover coin on the non-marked ones. The normalized uniform superposition of all vertices is the initial condition to avoid any bias at the beginning.  This method is called abstract search algorithm~\cite{Ambainis:2005,Portugal:book} and  was used for the hypercube~\cite{Shenvi:2003}, two-dimensional lattice~\cite{Ambainis:2005,Tulsi:2008}, honeycomb network~\cite{Abal:2010}, triangular network~\cite{Abal:2011}, and various graphs~\cite{BW10,LW12}. Let us review this method in the context of non-regular graphs. The coin is defined by
\begin{equation}
C'\ket{v,j}\,=\,\begin{cases} -\ket{v,j} &\mbox{if } v \mbox{ is a marked vertex} \\
-\ket{v,j} + 2\ket{\psi}& \mbox{if } v \mbox{ is not a marked vertex}, \end{cases}
\end{equation}
where $\ket{\psi}=\frac{1}{d_v}\sum_{j=0}^{d_v-1}\ket{v,j}$ (the uniform superposition of all coin directions at vertex $v$), $d_v$ is the degree of vertex $v$, and we are using the notation $\ket{v,j}$ following the one described in the paragraph right after Definition~\ref{def:nonregularQW}. The evolution operator is $U=S\,C'$ characterizing a non-regular flip-flop coined QW (Definition~\ref{def:nonregularQW}) even if graph $\Gamma(V,E)$, on which the QW takes place, is regular. Let $C\in {\cal H}^{2|E|}$ be the usual coin with the Grover operator $G$ for all vertices (including the marked ones). The searching evolution operator $U$ can be written as $(S\,C)\cdot R$, where $(S\,C)$ is the evolution operator of a QW with no marked vertices and $R$ is a reflection that applies $(-G)$ on the marked vertices and $(+I)$ on the non-marked ones, because $C'=CR$. By using this fact and the spectrum of $(S\,C)$, Ref.~\cite{Ambainis:2005} was able to find two non-trivial eigenvectors of $U$ associated with the eigenvalues with the smallest positive argument, which enabled the authors to find analytically the time complexity of the algorithm for the spatial search problem on the two-dimensional lattice with one marked vertex. Using the evolution operator $U=S\,C'$, the probability at the marked vertex increases periodically allowing to find a marked vertex if one performs the measurement at the correct moment.

On the other hand, the quantum search method in Szegedy's framework is an extension of the classical method using random walks. The key concept in the classical case is the \textbf{hitting time}, which is the average time to hit a marked vertex for the first time using a random walk on a graph $\Gamma$ with stochastic matrix $P$ after specifying some initial condition. The classical hitting time can be calculated by converting the original graph $\Gamma$ into a new directed graph $\Gamma'$ (with a new stochastic matrix $P'$) by removing the edges that leave the marked vertices. Marked vertices are converted into sinks. The hitting time obtained using $P'$ is the same using $P$ because as soon as the walker hits a marked vertex using some edge that comes from a non-marked vertex, the walker needs not to go ahead. Szegedy proposed a quantum version of this procedure. Let $\Gamma(X,E)$ be the original classical graph, where $X$ is the set of vertices and $E$ the set of edges. Define $\Gamma(X,X',E')$ as a bipartite graph obtained from  $\Gamma(X,E)$ by duplicating $X$ and by converting edges $\{x_i,x_j\}\in E$ into $\{x_i,x'_j\}\in E'$. Until this point, no vertex has been marked and the quantum walk described in Definition~\ref{def:SzegedyQW} can be used taking $P=Q$. As before, define $\Gamma'(X,X',E'')$ as a directed bipartite graph by removing the edges of $\Gamma(X,X',E')$ that leaves the marked vertices of $X$ and $X'$. Add new non-directed edges connecting a marked $x$ with its corresponding copy $x'$ for all marked vertices. The quantum walk described in Definition~\ref{def:SzegedyQW} can be used taking $P'=Q'$, where $P'$ and $Q'$ are the new stochastic matrices of $\Gamma'(X,X',E'')$. With this framework and taking 
\begin{equation}
\ket{\psi_0}=\frac{1}{\sqrt n}\sum_{x y} \sqrt{p_{x y}} \ket{x,y}
\end{equation}
as initial condition in ${\cal H}^{n}\otimes {\cal H}^{n}$,  Szegedy showed that the detection problem on $\Gamma'(X,X',E')$ can be solved with a quadratic speedup compared with the time complexity of the same problem using random walks with symmetric and ergodic stochastic matrix $P$ in $\Gamma(X,E)$~\cite{Szegedy:2004}. In the detection problem, one does not calculate the probability of finding a marked vertex and, therefore, cannot be sure to have found the marked vertex. The searching problem on \textit{bipartite graphs} with a single marked vertex was addressed in Ref.~\cite{KMOR15}.

Szegedy's searching framework can be straightforwardly extended for generic bipartite graphs $\Gamma(X,Y,E)$, which are not obtained from the duplication process of simple classical graphs, but instead are obtained from the duplication process of \textit{directed} classical graphs. The key ingredient is to use sinks, which need not to be in both sets $X$ and $Y$. Since the elements of $X$ represent the physical positions while the elements of $Y$ are auxiliary copies, a good strategy is to use sinks only in $X$. To mark vertices in $X$, define a new directed bipartite graph $\Gamma'(X,Y,E')$ by removing the edges of $\Gamma(X,Y,E)$ that leave the marked vertices of $X$; they become sinks. Use the new stochastic matrices $P'$ and $Q'$ of $\Gamma'(X,Y,E')$ and the corresponding vectors~(\ref{ht_phi_x}) and~(\ref{ht_psi_y}) to define Szegedy's searching QW on the new directed bipartite graph. $P'$ has complete rows of zeroes corresponding to the marked vertices and it is not a stochastic matrix in the usual sense. $Q'$ does not change, because we are introducing sinks only in $X$.  If the initial condition is a uniform superposition of the edges of $\Gamma(X,Y,E)$, Szegedy's QW will find a marked vertex in the sense that the probability associated with the marked vertices will be high if one performs a measurement (projection on the computational basis of $X$) at the correct time. We call this extended method \textbf{Szegedy's searching framework}.

Coin-based search algorithms and Szegedy's searching framework seem to be very different. However, they are strongly related. Let us address the equivalence between the abstract search algorithm and Szegedy's framework using the following strategy: (1)~we review how to convert Szegedy's QWs with sinks into equivalent generalized staggered QWs following~\cite{PSFG15}, and (2)~we show how to convert coined QWs with marked vertices into equivalent generalized staggered QWs. The generalized staggered models coming from those two different directions are equivalent, if the stochastic matrix of Szegedy's QW obeys the premises of Theorem~\ref{theo4} and vectors $\ket{\phi_x}$ and $\ket{\psi_y}$ are the uniform superposition.

Let us address item (1)~by reviewing how Szegedy's QWs with sinks are converted into equivalent generalized staggered QWs. Ref.~\cite{PSFG15} showed that Szegedy's searching framework is included in the staggered searching method, which employs partial tessellations. The staggered QW graph is the line graph of the original bipartite graph (before creating sinks). Tessellation $\alpha$ is partial; it does not employ polygons with the vertices corresponding to the edges that were removed in the process of creating the sinks. Following this process, there will be edges in the line graph that do not belong to neither partial tessellation $\alpha$ nor tessellation $\beta$. Tessellation $\beta$ is the same used in the proof of Theorem~\ref{theo4}. For example, the directed bipartite graph $\Gamma''$ of Fig.~\ref{fig:figure9a} has one sink: vertex $4\in X$. Graph $\Gamma'$ (with the dashed edges) is the line graph of the graph with no sinks equivalent to  $\Gamma''$. $\Gamma'$ has periodic boundary conditions (in the form of a torus). In the figure, we label the polygons of tessellation $\alpha$ (blue) from 0 to 8 (4 is missing) and tessellation $\beta$ (red) from 0 to 17. There is an one-to-one mapping between the polygons of tessellation $\alpha$ and the vertex labels of $\Gamma''$ in $X$ (the same with respect to tessellation $\beta$ and labels in $Y$). The polygon of tessellation $\alpha$ corresponding to vertex $4\in X$ must be missing. The staggered QW on $\Gamma'$ with partial tessellation $\alpha$ and complete tessellation $\beta$ is equivalent to Szegedy's QW on $\Gamma''$ with vertex $4\in X$ as a sink.

\begin{figure}[h!] 
\centering
\includegraphics[scale=0.80]{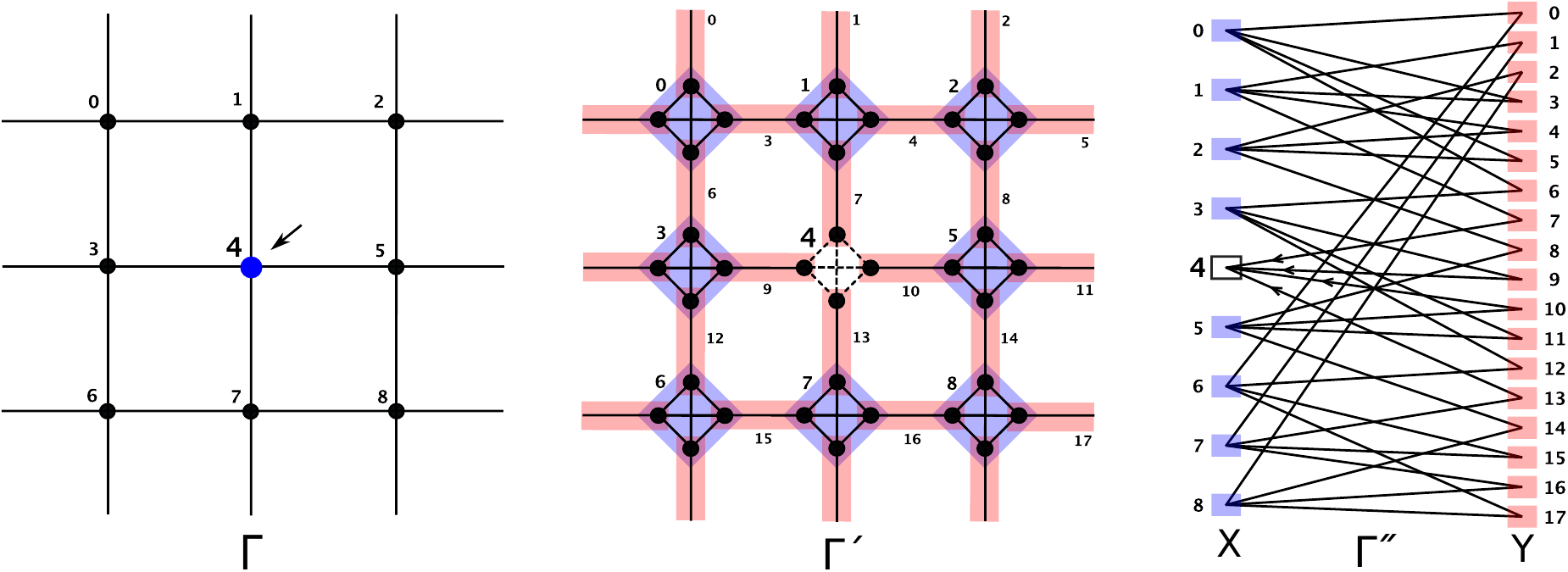}
\caption{A search algorithm using coined QW on a two-dimensional lattice $\Gamma$ with periodic boundary conditions using $(-I)$ on the marked vertex is equivalent to the generalized staggered QW on $\Gamma'$ with a missing polygon at the center of $\Gamma'$, which is equivalent to Szegedy's search on the directed bipartite graph $\Gamma''$ with vertex $4\in X$ as a sink.} 
\label{fig:figure9a}
\end{figure}

Let us address item (2). When we convert a non-regular flip-flop coined QW into an equivalent staggered QW, the vertices with the Grover coin are replaced by cliques while the vertices with coin $(-I)$ are converted into disconnected vertices (empty graphs) with no polygons because $(-I)$ has no $(+1)$-eigenvectors. In this case, tessellation $\alpha$ is partial. Tessellation $\beta$ is complete because there is no change in the shift operator. For example, the two-dimensional lattice $\Gamma$ with periodic boundary conditions of Fig.~\ref{fig:figure9a} with a marked vertex in the center (label-4 blue vertex with an arrow pointing to it) must be converted to graph $\Gamma'$ (without the dashed edges) in the middle of the figure. Notice that the graph and the tessellations obtained from $\Gamma$ coincide with the graph and tessellations obtained from the bipartite graph employed by Szegedy's searching model after removing the dashed edges.  The dashed edges can be removed because they play no role in the quantum-walk dynamics; they belong to no polygon. The non-regular flip-flop coined QW on $\Gamma$ with marked vertex 4 is equivalent to Szegedy's QW on $\Gamma''$ with vertex $4\in X$ as a sink in the sense that the evolution operators are exactly the same if we eliminate the idle space in the Szegedy's evolution operator and choose bases in the proper ordering. Notice that $\Gamma$ is not the classical graph associated with $\Gamma''$. The classical graph has 17 vertices and is a directed absorbing graph (the sink is vertex 4).

\begin{theorem}\label{theo5}
A non-regular flip-flop coined QW on a graph $\Gamma(V,E)$ with coin $(-I)$ on the marked vertices and the Grover coin on the non-marked vertices can be cast into Szegedy's searching framework.
\end{theorem}

\begin{theorem}\label{theo6}
Let $\Gamma(X,Y,E)$ be a bipartite graph such that $\deg(y)=2$, $\forall y \in Y$. Suppose that $q_{yx}$ is either $1/2$ or $0$. Then Szegedy's QW on $\Gamma(X,Y,E)$ with sinks in set $X$ is equivalent to a non-regular flip-flop coined QW on a $|X|$-multigraph $\Gamma(V,E')$ with coin $(-I)$ on the vertices $v\in V$ associated with the sinks of $X$. 
\end{theorem}
The proofs are in the Appendix.

\section{Conclusions}\label{sec:conc}

In this work, we have showed that the coined and Szegedy's models have in common a large class of QWs. Under some assumptions, we can convert the graph on which the coined QW takes place into a bipartite graph on which an equivalent Szegedy's QW takes place, and vice versa. The equivalence means that the QW state at any step $t$ of one model can be exactly obtained using the other model. To go from the coined model on graph $\Gamma$ to Szegedy's model on graph $\Gamma''$, we have to replace each vertex of $\Gamma$ by a clique (for the Grover coin) obtaining a new enlarged graph $\Gamma'$ on which an equivalent staggered QW is defined. As a last step, we have to find the bipartite graph $\Gamma''$, the line graph of which is (isomorphic to) $\Gamma'$. On the other direction, we start with a bipartite graph $\Gamma''$ on which Szegedy's QW takes place, and we have to find the line graph of $\Gamma''$, convert cliques of one tessellation into vertices, which generates a new graph, or multigraph in some cases, on which the coined QW takes place. The staggered QW model plays a key role in the conversion process, because Szegedy's model has an idle subspace which hinders the direct conversion from Szegedy's to the coined model.

Remarkably, the abstract search algorithm using the coined QW model on (non-regular) graphs can be cast into Szegedy's searching framework, which is based on bipartite graphs with sinks. When converting from the coined to the staggered model, the coin $(-I)$ represents a missing polygon in one tessellation, which is converted into a sink in the equivalent Szegedy's model. The process is true on the other way around under some restrictions on the stochastic matrix of the bipartite graph; Szegedy's QWs with sinks can be converted into an equivalent search algorithm in the coined model using $(-I)$ on the marked vertices. One restriction is the degree of the vertices of set $Y$ must be 2. Szegedy's QWs that do not obey this restriction are not equivalent to coined QWs. In this sense, Szegedy's model is more general. On the other hand, coined QWs using coins that are not reflections, such as the Fourier coin~\cite{Portugal:book}, cannot be cast into Szegedy's model. In this sense, coined QW model is more general.

In conclusion, Szegedy's and the coined QW models share a large class of QW instances. However, there are Szegedy's QWs that cannot be converted into the coined formalism using the standard and non-regular flip-flop coined QWs defined in Sec.~\ref{sec:MD}, and vice versa. Since Szegedy's model is a subset of the staggered model, we also conclude that coined and staggered models are not equivalent. To pursue further connections between Szegedy's (or staggered) and coined models, one has to generalize the definitions of those models.

As a byproduct, we have showed how to convert coined QWs into coinless QWs in an enlarged graph, when the coin is an orthogonal reflection. The coin becomes a unitary operator acting on the new vertices and edges of the extended graph. The coined model can be understood in a new way, which may help in experimental implementations or in decoherence analysis.

\section*{Acknowledgements}
The author acknowledges financial support from Faperj (grant n.~E-26/102.350/2013) and CNPq (grants n.~304709/2011-5, 4741\-43/2013-9, and 400216/2014-0). The author thanks useful discussion with Raqueline A.M.~Santos, Tharso D.~Fernandes, Stefan Boettcher, Andris Ambainis, and the quantum computing group of LNCC.

\section*{Appendix}

\subsection*{Proof of Theorem~\ref{theo3}}

Suppose that we have a well-defined non-regular flip-flop coined QW on a graph $\Gamma(V,E)$ with coin $C'$ being an orthogonal reflection. As an intermediate step, we use a staggered QW on a graph $\Gamma(V',E')$ with two tessellations equivalent to the coined QW. If $C'=C_1\oplus ... \oplus C_{|V|}$, the eigenvectors of $C'$ are the direct sum of eigenvectors of $C_v$ and zero vectors, for $1\le v\le |V|$. Let $\ket{\tilde\alpha_x^{(v)}}$, $0\le x < m_v$ be an orthonormal basis for the $(+1)$-eigenspace of $C_v$ and let $\ket{\alpha_{x_v}^{(v)}}$ be the corresponding eigenvectors of $C'$ obtained from $\ket{\tilde\alpha_{x_v}^{(v)}}$ by performing the necessary direct sums with zero vectors. If $C'$ is an orthogonal reflection of graph $\Gamma'(V',E')$, it can be written as
\begin{eqnarray}\label{Cprime2}
  C' &=& 2\,\sum_{v=1}^{{|V|}} \sum_{x_v=0}^{m_v-1} \ket{\alpha_{x_v}^{(v)}}\bra{\alpha_{x_v}^{(v)}} - I,
\end{eqnarray}
where the set of $(+1)$-eigenvectors $\ket{\alpha_{x_v}^{(v)}}$ has the following properties: (1)~if the $i$-th entry of $\ket{\alpha_{x_v}^{(v)}}$ is nonzero, the $i$-th entries of the other $(+1)$-eigenvectors must be zero, and (2)~vector $\sum_{v=1}^{{|V|}}\sum_{x_v=0}^{m_v-1} \ket{\alpha_{x_v}^{(v)}}$ has no zero entries. Then each $C_v$ can be written as
\begin{eqnarray}
  C_v &=& 2\sum_{x=0}^{m_v-1} \ket{\tilde\alpha_x^{(v)}}\bra{\tilde\alpha_x^{(v)}} - I \label{Cv}
\end{eqnarray}
and the set of vectors $\ket{\tilde\alpha_x^{(v)}}$ inherit properties~(1) and~(2)~in ${\cal H}^{d_v}$. Each $C_v$ is an orthogonal reflection in ${\cal H}^{d_v}$, where $d_v$ is the degree of vertex $v$ in $\Gamma(V,E)$, and has an associated  $d_v$-graph $\Gamma_{C_v}$ tessellated by the $(+1)$-eigenvectors of $C_v$. $\Gamma_{C_v}$ is a union of $m_v$ disjoint cliques. If $C_v$ has only one $(+1)$-eigenvector ($m_v=1$), $\Gamma_{C_v}$ is a clique.

Graph $\Gamma'(V',E')$ is obtained from $\Gamma(V,E)$ by replacing each vertex $v\in V$ by graph $\Gamma_{C_v}$  gluing the vertices of $\Gamma_{C_v}$, which run from 0 to $d_v-1$, in one-to-one mapping with the labels of the coin directions at vertex $v$. The vertices of $\Gamma_{C_v}$ after the gluing process receive the labels of the basis vectors in $\ket{\alpha_{x_v}^{(v)}}$ with nonzero coefficients. For example, the vertices of the 3-clique in graph $\Gamma'$ of Fig.~\ref{fig:example6} has labels 1, 2, and 3 because they correspond to the basis vectors of the $(+1)$-eigenvector $\ket{\alpha_1}=(\ket{1}+\ket{2}+\ket{3})/\sqrt 3$. 

Polygons induced by $\ket{\alpha_x^{(v)}},\,\forall v,x$ tessellate $\Gamma'(V',E')$ because polygons induced by $\ket{\alpha_{x}^{(v)}}$, $0\le x<m_v$ exactly cover the graphs $\Gamma_{C_v}$ that replace vertices $v$. Tessellation $\alpha$ covers all vertices and all edges that were added via  $\Gamma_{C_v}$, for all $v$. This tessellation does not cover the edges of $\Gamma'(V',E')$ that were inherited from the original graph  $\Gamma(V,E)$. 

Tessellation $\beta$ is made of size-2 polygons that cover the edges of $\Gamma'(V',E')$ that were inherited from the original graph  $\Gamma(V,E)$. This tessellation has $|E|$ polygons and the set of those polygons has an one-to-one mapping with an independent set of $(+1)$-eigenvectors of $S$ in the computational basis,  which are given by vectors
\begin{equation}
   \ket{\beta_{v}^{j}} \,=\, \frac{1}{\sqrt 2}\left(\ket{v,j}+\ket{v',j'}\right)
\end{equation}
using the notation of Eq.~(\ref{def_S2}), where $v\in V$ and $0\le j\le d_v-1$. The cardinality of the independent set of $(+1)$-eigenvectors of $S$ is $|E|$. The shift operator of the non-regular flip-flop coined QW is
\begin{equation}
	S\,=\, 2 \sum \ket{\beta_{v}^{j}}\bra{\beta_{v}^{j}}-I
\end{equation}
where the sum runs over the set of independent $(+1)$-eigenvectors (the sum has $|E|$ terms). $S$  is an orthogonal reflection because the set of independent $(+1)$-eigenvectors has non-overlapping nonzero entries and the sum of those eigenvectors has no zero entries in the computation basis of $\Gamma'(V',E')$. Polygons induced by $\ket{\beta_{v}^{j}}$ form a perfect matching of $\Gamma'(V',E')$. 

The union of tessellations $\alpha$ and $\beta$ covers all edges and is a well-defined staggered QW having one vertex in each polygon intersection. Using Proposition~4.3 of Ref.~\cite{PSFG15}, this staggered QW can be cast into the extended Szegedy's framework.\qed

\subsection*{Proof of Theorem~\ref{theo4}}

Consider the staggered QW model on the line graph $L(\Gamma)$ equivalent to Szegedy's QW on $\Gamma(X,Y,E)$. $L(\Gamma)$ has $2|Y|$ vertices. The polygons of the staggered model are induced by
\begin{eqnarray}
  \ket{\alpha_x} &=&  \sum_{y\in Y} \sqrt{p_{x y}} \, \ket{f(x,y)}, \label{ht_alpha_x2} \\
  \ket{\beta_y}  &=&  \sum_{x\in X} \sqrt{q_{y x}}\, \ket{f(x,y)}, \label{ht_beta_y2}
\end{eqnarray} 
where $f$ is the bijection between $E$ and the vertices of $L(\Gamma)$ as described in Ref.~\cite{PSFG15} and vectors $\ket{\alpha_x},\ket{\beta_x}$ belong to Hilbert space ${\cal H}^{2|Y|}$. 

Tessellation $\alpha$ is induced by the orthogonal reflection
\begin{eqnarray}
  U_0 &=& 2\sum_{x\in X} \ket{\alpha_x}\bra{\alpha_x} - I. 
\end{eqnarray}
By using a proper choice of $f$, matrix $\ket{\alpha_x}\bra{\alpha_x}$ is a direct sum of zeros matrices and a $d_x\times d_x$ matrix $M_x$, which has no zero entries. Define 
\begin{equation}
	C_x\,=\, 2 M_x - I.
\end{equation}
Then 
\begin{equation}\label{coin_C_Sz2}
	U_0 \,=\, \bigoplus_{x\in X} C_x.
\end{equation}

Operator $U_1$ is equal to the one described in the proof of Theorem~\ref{theo2} and is given by Eq.~(\ref{U_1_theo2}) because the assumptions about vertices $y\in Y$ are equal to the ones in Theorem~\ref{theo2}. Then $U_1$ commutes basis vectors and $U_1^2=I$. 

The evolution operator is
\begin{equation}
U \,=\, U_1 \,U_0,
\end{equation}
where $U_0$ is the coin operator given by Eq.~(\ref{coin_C_Sz2}) and $U_1$ is the shift operator given by Eq.~(\ref{U1_theo2}). $U$ is an evolution operator of a non-regular flip-flop coined QW on the (multi)graph obtained in the following way: The polygons of tessellation $\beta$ have two vertices and form a perfect matching of $L(\Gamma)$. The remaining cliques belong to tessellation $\alpha$. Each clique of tessellation $\alpha$ must be converted into a single vertex. If two cliques of tessellation $\alpha$ are connected by an edge, the vertices that replace those cliques are adjacent.  If two cliques are connected by more than one edge, the vertices that replace those cliques must be connected by more than one edge generating a non-regular $|X|$-multigraph. \qed

\subsection*{Proof of Theorem~\ref{theo5}}

\begin{proof}
The method employed in the proof of Theorem~\ref{theo3} when the coin is an orthogonal reflection can be straightforwardly extended  when the coin is a \textit{partial} orthogonal reflection. In this case, we can convert a non-regular flip-flop coined QW on a graph $\Gamma(V,E)$ with coin $(-I)$ on the marked vertices and the Grover coin on the non-marked vertices into an equivalent \textit{generalized} staggered QW on $\Gamma'(V',E')$, which is obtained from  $\Gamma(V,E)$ in the following way: a non-marked vertex $v\in V$ is converted into $d_v$-cliques and a marked vertex $v$ into disconnected $d_v$-graphs (empty $d_v$-graphs), where $d_v$ is the degree of vertex $v$. Tessellation $\alpha$ is partial, with polygons being the cliques associated with non-marked vertices only. Tessellation $\beta$ is the same employed in the proof of Theorem~\ref{theo3}.

The next step is to define a new graph $\Gamma''(V',E'')$ by converting the empty $d_v$-graphs into complete graphs by adding new edges to $\Gamma'(V',E')$. Let $\tilde\alpha$ be an extension of partial tessellation $\alpha$ by adding new polygons corresponding to the new complete graphs.   $\Gamma''(V',E'')$ is the line graph of some bipartite graph $\Gamma(X,Y,\tilde E)$ because the union of tessellations $\tilde\alpha$ and $\beta$ form a two-colorable Kraus partition of $\Gamma''(V',E'')$. 

We have defined a generalized staggered QW on $\Gamma'(V',E')$, which is equivalent to a generalized staggered QW on $\Gamma''(V',E'')$ using partial tessellation $\alpha$ because the egdes in $E''\setminus E'$ do not belong to any polygon. Following Ref.~\cite{PSFG15} we can obtain an equivalent Szegedy's QW; the missing polygons in partial tessellation $\alpha$ create sinks in graph $\Gamma(X,Y,\tilde E)$ by removing the directed edges coming out of the vertices in $X$ associated with the missing polygons. The edges oriented to the sinks are kept. This process creates a new directed bipartite graph $\Gamma'(X,Y,\tilde E')$. Ref.~\cite{PSFG15} showed that Szegedy's QW on $\Gamma'(X,Y,\tilde E')$ with vectors $\ket{\phi_x}$ and $\ket{\psi_y}$ given by Eqs.~(\ref{ht_phi_x}) and~(\ref{ht_psi_y}) in uniform superposition is equivalent to the generalized staggered QW on $\Gamma'(V',E')$. Then the non-regular flip-flop coined QW on a graph $\Gamma(V,E)$ with coin $(-I)$ on the marked vertices and the Grover coin on the non-marked vertices can be cast into Szegedy's searching framework.  \qed
\end{proof}

\subsection*{Proof of Theorem~\ref{theo6}}
\begin{proof}
This theorem is a corollary of Theorem~\ref{theo4}, if we employ the method described in Ref.~\cite{PSFG15} of converting Szegedy's QWs on bipartited graphs with sinks into generalized staggered QWs. To convert generalized staggered QWs into the coined QW model, a missing polygon is converted into coin $(-I)$. \qed
\end{proof}


\end{document}